\documentclass[a4paper,11pt,reqno]{amsart}

\usepackage{amssymb,mathptmx}
\usepackage{microtype,cite}
\usepackage{booktabs}
\usepackage{tikz}

\usepackage[colorlinks,linkcolor=blue,citecolor=blue,urlcolor=blue]{hyperref}
\newtheorem{lemma}{Lemma}

\theoremstyle{definition}

\theoremstyle{remark}
\newtheorem{remark}{Remark}

\DeclareMathOperator{\PP}{\mathbb{P}}
\DeclareMathOperator{\EE}{\mathbb{E}}

\DeclareMathOperator{\Beta}{\ensuremath{Beta}}
\DeclareMathOperator{\Bernoulli}{\ensuremath{Bernoulli}}

\newcommand{\iid}{i.{\kern1pt}i.{\kern1pt}d.}
\newcommand{\ie}{i.{\kern1pt}e.}

\renewcommand{\mid}{\mkern2mu|\mkern2mu}

\renewcommand{\leq}{\leqslant}
\renewcommand{\geq}{\geqslant}
\newcommand{\eps}{\epsilon}

\newcommand{\deq}{\mathrel{\mathop:}=}

\newcommand{\mku}{\mkern1mu}
\renewcommand{\,}{\ifmmode\mkern2mu\else\thinspace\fi}

\newcommand{\XN}{X_{N}}

\begin{document}

\title[Finite-depth scaling and a Bernoulli-leaf identity for the min-plus process]{Finite-depth scaling and an exact Bernoulli-leaf identity for the min-plus process on the binary tree}

\author[J. R. G. Mendon\c{c}a]{Jos\'{e} Ricardo G. Mendon\c{c}a}
\address{Universidade de S\~{a}o Paulo, S\~{a}o Paulo, SP, Brasil}
\email{jricardo@usp.br}


\begin{abstract}
The min-plus process is a stochastic coagulation-annihilation-type process on the binary tree, of interest in mathematics, physics, and computer science as a tractable instance of max-type recursive distributional equations. We carry out large Monte Carlo simulations at effective tree depths up to $N=60$ that provide finite-depth corroboration of the $\Beta(2,1)$ stretched-exponential limit for its root value $\XN$ at $p=1/2$, on the asymmetric $\sqrt{N}$ side of the random-homogeneous-systems classification recently introduced by Chen, Duquesne, and Shi and by Morfe. Off criticality, our simulations confirm the sub-critical closed form $\PP(X_{\infty}=1)=(1-2p)/(1-p)$ within Monte Carlo error and document a super-critical mean growth exceeding the elementary $(2p)^{N}$ lower bound at the depths we reach. For a $\Bernoulli(q)$-initial-condition variant, we identify an elementary closed-form identity at $p=1/2$ that pins down the order parameter $\PP(\XN=0) = q$ exactly, locates the absorbing-state phase transition at $p_{c}=1/2$ in the operator-mixing probability rather than in the initial-zero density, and shows that the conditional law on positives deforms substantially with $q$. Our simulations use a level-wise recursion and an FFT-based precomputed leaf table which reduce the effective simulation depth while preserving the recursive tree law and may be useful for the simulation of related recursive equations on large trees.
\end{abstract}

\subjclass[2020]{Primary: 60J80; Secondary: 60K35, 82B20, 65C05}

\keywords{min-plus process, max-type recursive distributional equation, absorbing-state phase transition, hipster random walk, stretched-exponential distribution, Monte Carlo simulation}

\maketitle


\section{Introduction}
\label{intro}

Let $X_{0}^{N} \deq (X_{0}(1), \ldots, X_{0}(2^{N}))$ be a family of nonnegative integers attached to the labeled leaves of a rooted binary tree $T_{N}$ of depth $N \geq 1$ with $|T_{N}| = 2^{N+1}-1$ vertices. We define the min-plus stochastic process on $T_{N}$ as follows~\cite{AuffingerCable2017}: with the leaves indexed by $n=0$, the parent $X_{n+1}(i)$ of the nodes $X_{n}(2i-1)$, $X_{n}(2i)$, $i=1, \dots, 2^{N-n-1}$, takes the value
\begin{equation}
\label{eq:minplus}
X_{n+1}(i) = \begin{cases}
X_{n}(2i-1)+X_{n}(2i) & \text{with probability } p, \\
\min(X_{n}(2i-1),X_{n}(2i)) & \text{with probability } 1-p,
\end{cases}
\end{equation}
where $\min(a,b)=\frac{1}{2}(a+b-|a-b|)$ and $0 \leq p \leq 1$. The process resembles a stochastic coagulation-annihilation process moving up the tree toward the root; see Figure~\ref{fig:minplus}. The question is: what is the distribution of $\XN \deq \XN(1)$ at the root in the limit $N \to \infty$? 

The min-plus process was raised by R.~Pemantle as an open problem in 2017 and given its first rigorous treatment by Auffinger and Cable~\cite{AuffingerCable2017}. The case $p=1/2$ lies exactly at the percolation threshold of the underlying binary tree of `min'-labeled internal nodes (cf.~Remark~\ref{rmk:percolation}), and is the regime of greatest interest.

\begin{figure}[t]
\centering
\begin{tikzpicture}[
  level distance=10mm,
  level 1/.style={sibling distance=64mm},
  level 2/.style={sibling distance=32mm},
  level 3/.style={sibling distance=16mm},
  level 4/.style={sibling distance=8mm},
  minnode/.style={rectangle, draw=black, fill=white, text=black,
                  minimum size=5mm, inner sep=0pt, font=\scriptsize},
  plusnode/.style={circle, draw=black, fill=white, text=black,
                   minimum size=5.64mm, inner sep=0pt, font=\scriptsize},
  leaf/.style={circle, draw=black, densely dashed, fill=white, text=black,
               minimum size=5.64mm, inner sep=0pt, font=\scriptsize}
]
\node[plusnode] {3}
  child {node[minnode] {1}
    child {node[plusnode] {2}
      child {node[minnode] {1}
        child {node[leaf] {1}} child {node[leaf] {1}} }
      child {node[minnode] {1}
        child {node[leaf] {1}} child {node[leaf] {1}} } }
    child {node[minnode] {1}
      child {node[minnode] {1}
        child {node[leaf] {1}} child {node[leaf] {1}} }
      child {node[minnode] {1}
        child {node[leaf] {1}} child {node[leaf] {1}} } } }
  child {node[plusnode] {2}
    child {node[minnode] {1}
      child {node[minnode] {1}
        child {node[leaf] {1}} child {node[leaf] {1}} }
      child {node[plusnode] {2}
        child {node[leaf] {1}} child {node[leaf] {1}} } }
    child {node[minnode] {1}
      child {node[plusnode] {2}
        child {node[leaf] {1}} child {node[leaf] {1}} }
      child {node[minnode] {1}
        child {node[leaf] {1}} child {node[leaf] {1}} } } };
\end{tikzpicture}
\caption{\label{fig:minplus}A realization of the critical min-plus stochastic process on the binary tree $T_{4}$. Squares are internal nodes labeled $\min$; circles are internal nodes labeled $+$; dashed circles at the bottom are the leaves, here all set to $X_{0}(i) = 1$. Values inside each node are the level-by-level outputs of the recursion (\ref{eq:minplus}). The realization shown gives root value $\XN = 3$.}
\end{figure}

The min-plus process belongs to the family of {max-type recursive distributional equations} of the form $X \stackrel{d}{=} A(X^{(1)}, X^{(2)}, \ldots, X^{(b)})$, where the $X^{(i)}$ are independent copies of $X$ and $A$ is a (possibly random) operator~\cite{AldousBandyopadhyay2005}, and is among the simplest tractable instances in which the operator itself is selected at random between two competing alternatives at every node of the tree. The $\Beta(2,1)$ limit law identified by \cite{AuffingerCable2017} for $\XN$ at $p=1/2$ also governs the lazy hipster random walk of~\cite{AddarioBerryCairns2020}, placing both in an asymmetric $\sqrt{N}$ universality class alongside a companion symmetric class at rate $N^{1/3}$. See Remark~\ref{rmk:hipster} for this classification and where the min-plus process lies within it.

A related class of recursive distributional problems appears in random Boolean formulas. Any Boolean function on $n$ variables can be represented as a labeled binary tree with literals at the leaves and the connectives \textbf{and} and \textbf{or} at internal nodes, and the asymptotic statistics of the function computed by a random \textbf{and}--\textbf{or} tree---the probability that a random formula of depth $N$ computes a given function, or the typical formula size required to compute it---are recursive distributional problems with the multiplicative pair $\{\textbf{and},\textbf{or}\} \mapsto \{xy,\, 1-(1-x)(1-y)\}$ in place of $\{\min,+\}$~\cite{PemantleWard2006,LefmannSavicky1997,GardyWoods2005,ChauvinFlajolet2004}. Lindenstrauss and Talagrand~\cite{LindenstraussTalagrand2022} have recently shown that random balanced \textbf{and}--\textbf{or} trees of depth $N$ become constant (tautology or antitautology) above the stretched-exponential threshold $k_{N} \sim \exp(\sqrt{N})$ in the number of input variables---the same $\sqrt{N}$ critical scale as the min-plus root value.

Similar recursive structures occur in hierarchical renormalization models for pinning transitions in disordered systems, going back to Cook and Derrida's 1989 study of polymers on hierarchical lattices~\cite{CookDerrida1989}. The most studied modern incarnation is the Derrida--Retaux toy model~\cite{DerridaRetaux2014} of the depinning of an interface from a disordered substrate~\cite{DerridaHakimVannimenus1992}, equivalent after shifting and rescaling to the canonical max-type iteration $X_{n+1} = \max(X_{n}^{(1)}+X_{n}^{(2)}-1,0)$ studied in~\cite{HuShi2018,ChenDagard2021,DerridaShi2020}. Derrida and Retaux conjectured an infinite-order phase transition of Berezinskii--Kosterlitz--Thouless type for the free energy of this recursion as a function of the leaf distribution; the BKT exponent $1/2$ has since been established under suitable integrability conditions~\cite{ChenDagard2021}, while different integrability conditions on the leaves give rise to a family of distinct universality classes~\cite{HuShi2018}. Several open questions are shared with the min-plus model; in particular, the derivation of the full distribution of $\XN$ would furnish the free energy for some of the physical models.

In this paper we report large-scale Monte Carlo simulations of the distribution of $\XN$ in light of the rigorous picture above. Section~\ref{sec:critical} treats the critical model ($p=1/2$, all $X_{0}(i)=1$) on trees of effective depth up to $N=60$, testing the data against the scaling law of~\cite{AuffingerCable2017}. Section~\ref{sec:noncritical} covers the off-critical regimes $p \neq 1/2$: the sub-critical sweep verifies the closed form $\PP(X_{\infty}=1) = (1-2p)/(1-p)$ within Monte Carlo error, while the super-critical sweep shows the lower bound $\EE[\XN] \geq (2p)^{N}$ of Lemma~\ref{lem:px1}\,(iii) is not tight at accessible depths. Section~\ref{sec:variant} extends the model to $\Bernoulli(q)$ leaves, where an elementary recursion determines the order parameter $\PP(\XN=0)$ at $p=1/2$ and locates the absorbing-state transition at $p_{c} = 1/2$ in the operator-mixing probability rather than in the leaf-zero density. Section~\ref{sec:summary} summarizes and discusses perspectives on min-plus-type models.


\section{The critical min-plus process}
\label{sec:critical}

When all $X_{0}(i)=1$ and $p=1/2$, we refer to process (\ref{eq:minplus}) as the critical min-plus process. Although this section treats this critical case, it is economical to record the behavior of $\PP(X_{n}=1)$ for all $p$ at once; parts~(i) and~(iii) of the following lemma are used in Section~\ref{sec:noncritical}.

\begin{lemma}
\label{lem:px1}
In the min-plus process on $T_{N}$ with all $X_{0}(i) = 1$ and operator-mixing probability $p$, let $y_{n} \deq \PP(X_{n} = 1)$. Then $y_{n}$ obeys the deterministic recursion
\begin{equation}
\label{eq:y-rec}
y_{n+1} = (1-p)(2 y_{n} - y_{n}^{2}), \quad y_{0} = 1,
\end{equation}
with fixed points $y = 0$ and, for $p \leq 1/2$, $y^{*} = (1-2p)/(1-p) \in [0,1]$. Three regimes follow:
\begin{enumerate}
\item[(i)] For $p < 1/2$, $y_{n} \to y^{*} = (1-2p)/(1-p)$ at geometric rate $(2p)^{n}$.
\item[(ii)] At $p=1/2$, $y_{n} \to 0$ with $y_{n} \sim 2/n$.
\item[(iii)] For $p > 1/2$, $y_{n} \to 0$ at geometric rate $(2(1-p))^{n}$.
\end{enumerate}
\end{lemma}
\begin{proof}
Conditioning on the operator at the parent of $X_{n}^{(1)}$, $X_{n}^{(2)}$,
\[
\PP(X_{n+1} = 1) = p\PP(X_{n}^{(1)} + X_{n}^{(2)} = 1) + (1-p)\PP(\min(X_{n}^{(1)}, X_{n}^{(2)}) = 1).
\]
Since $X_{0} = 1$, induction on $n$ shows $X_{n} \geq 1$, so $X_{n}^{(1)} + X_{n}^{(2)} \geq 2$ and the first term vanishes. The minimum equals $1$ iff at least one of its arguments does, with probability $1 - (1-y_{n})^{2} = 2y_{n}-y_{n}^{2}$, yielding (\ref{eq:y-rec}). The fixed-point equation $y = (1-p)(2y-y^{2})$ factorizes as $y[(2p-1) + (1-p)\,y] = 0$, giving $y=0$ and $y=(1-2p)/(1-p)$, the latter in $[0,1]$ only for $p \leq 1/2$. The geometric rates in (i) and (iii) follow from $f'(y) = (1-p)(2-2y)$ evaluated at the relevant fixed point: $f'(0) = 2(1-p)$ and $f'(y^{*}) = 2p$. For (ii), at $p=1/2$ the recursion reads $y_{n+1} = y_{n} - y_{n}^{2}/2$; setting $z_{n} = 2/y_{n}$ gives $z_{n+1} = z_{n} + 1 + O(y_{n}) = z_{n} + 1 + O(1/n)$, so $z_{n} \sim n$ and $y_{n} \sim 2/n$.
\end{proof}

\begin{remark}
Lemma~\ref{lem:px1}\,(i) above recovers Lemma~16 of~\cite{AuffingerCable2017} for the limiting mass of $X_{\infty}$ at $1$.
\end{remark}

\begin{remark}
\label{rmk:percolation}
The recursion (\ref{eq:y-rec}) has a percolation interpretation: $\{X_{n} = 1\}$ is the event that there exists a path from the level-$n$ vertex down to a leaf along which every internal node is labeled `$\min$', equivalently the survival event of the Galton--Watson process whose offspring distribution attaches to each `alive' node a deterministic pair of alive nodes when its operator is `$\min$' and none when it is `$+$': $\xi = 2$ with probability $1-p$, $\xi = 0$ with probability $p$. The mean offspring $2(1-p)$ makes this branching process sub-critical for $p > 1/2$, critical at $p=1/2$, and super-critical for $p < 1/2$, matching the three regimes (i)--(iii). The leading constant in the $y_{n} \sim 2/n$ asymptotic at criticality is then Kolmogorov's classical theorem for critical Galton--Watson processes~\cite{AthreyaNey1972}, with offspring variance $\sigma^{2} = 4(1-p)p = 1$ at $p=1/2$, giving $n \cdot \PP(\text{survival to gen.}~n) \to 2/\sigma^{2} = 2$. The recursive structure places this transition within Johnson's taxonomy of recursive properties on Galton--Watson trees~\cite{Johnson2022}, where the criteria $f'(0) = 1$ and $f''(0) < 0$ at $p_{c} = 1/2$ mark it as continuous.
\end{remark}

We simulated the critical min-plus process on trees of effective depth $N$ ranging from $20$ to $60$ (i.\,e., up to $\approx 2.3 \times 10^{18}$ vertices), with $M=10^{5}$ Monte Carlo samples per depth. The simulation algorithm, the precomputed-leaves machinery that makes the deeper runs feasible, and the CPU and statistical bottlenecks are described in Appendix~\ref{sec:methods}.

The empirical distribution of $\XN$ for the largest accessible effective depths is shown in Figure~\ref{fig:dist_critical}, with the corresponding summary statistics collected in Table~\ref{tab:critical_summary}. The distribution is heavy-tailed and concentrates on small integers: most realizations have $\XN \in \{1,2,3\}$, while the maximum observed across $M=10^{5}$ samples increases steadily with $N$. The relevant rescaling is logarithmic: by \cite[Theorem~1]{AuffingerCable2017}, at $p=1/2$ one has
\begin{equation}
\label{eq:auffinger}
\frac{\log{\XN}}{\sqrt{cN}} \xrightarrow{~d~} \mathcal{B}, \quad c = \frac{\pi^{2}}{3},
\end{equation}
where $\mathcal{B} \sim \Beta(2,1)$, supported on $[0,1]$ with cdf $\PP(\mathcal{B} \leq t) = t^{2}$. Equivalently, $\log{\XN} = \Theta(\sqrt{N})$ in distribution, so $\XN$ grows at a {stretched-exponential} rate. The corresponding theoretical asymptote is \cite[Corollary~1]{AuffingerCable2017}
\begin{equation}
\label{eq:AClim}
\frac{\EE[\log{\XN}]}{\sqrt{N}} \to \frac{2\pi}{3\sqrt{3}} \approx 1.209,
\end{equation}
which we test against our data in Figure~\ref{fig:critical_scaling}. The empirical ratio $\widehat{\EE}[\log{\XN}]/\sqrt{N}$ rises monotonically from $0.745$ at $N=20$ to $0.870$ at $N=60$, well below the limit (\ref{eq:AClim}) even at $N=60$. Plotted against $1/\sqrt{N}$ (right panel of Figure~\ref{fig:critical_scaling}) the data fall along an approximately straight line. An ordinary least-squares fit to the seven depths $N = 36, 40, \ldots, 60$ gives
\begin{equation}
\frac{\widehat{\EE}[\log{\XN}]}{\sqrt{N}} \approx 1.060 - \frac{1.476}{\sqrt{N}},
\end{equation}
with $R^{2} > 0.994$; the intercept undershoots the limit (\ref{eq:AClim}) by $\approx 12\%$. The curvature visible in the right panel of Figure~\ref{fig:critical_scaling} is statistically significant, however, and once subleading structure is allowed the data become quantitatively consistent with the theoretical limit at the few-percent level. We defer the comparison of linear, quadratic, and Bulirsch--Stoer rational extrapolations to Appendix~\ref{sec:appendix-statistics} and Table~\ref{tab:extrap}.

\begin{figure}[t]
\centering
\includegraphics[width=\linewidth]{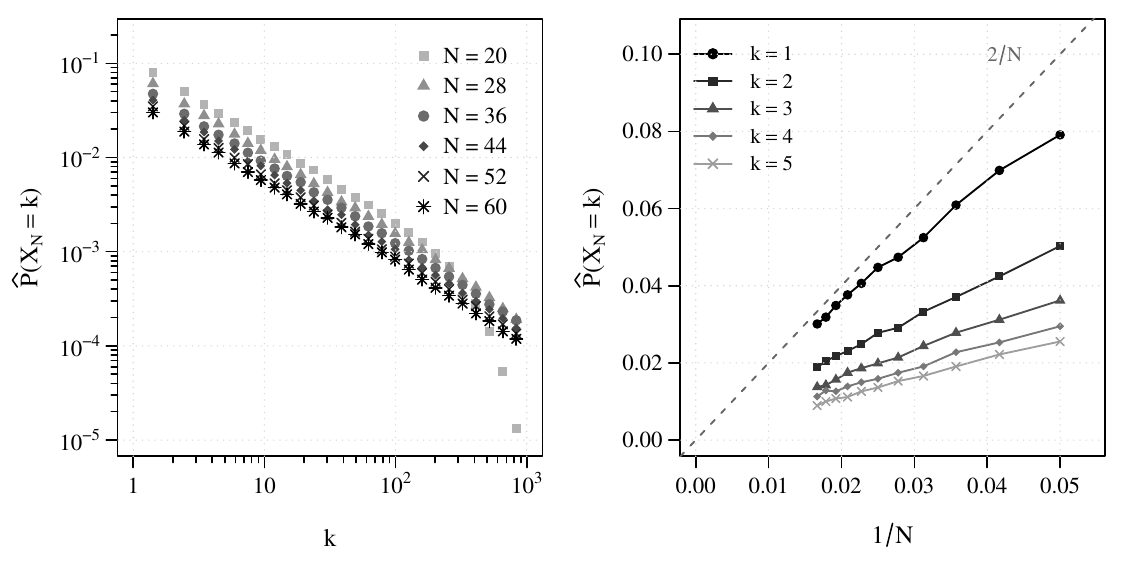}
\caption{\label{fig:dist_critical}\emph{Left:} empirical distribution $\widehat{\PP}(\XN = k)$ at the root of the critical min-plus process for a representative subset of effective depths, $M=10^{5}$ samples per $N$. Points are aggregated over geometric (log-spaced) bins and plotted at each bin's geometric midpoint; the mass at small integers shrinks with $N$ while the tail extends. \emph{Right:} the same data read at fixed $k$: $\widehat{\PP}(\XN = k)$ versus $1/N$ for $k = 1, \ldots, 5$ over all eleven simulated depths, with the theoretical rate $2/N$ of Lemma~\ref{lem:px1}\,(ii) dashed. No fixed $k$ retains mass as $N \to \infty$.}
\end{figure}

Beyond the leading mean, the full {shape} of the rescaled distribution at finite $N$ matches the theoretical $\Beta(2,1)$ profile once the prefactor lag is accounted for. Figure~\ref{fig:qq_beta21} compares the empirical quantile function of $\log{\XN}/\sqrt{cN}$ at $N=60$ against the $\Beta(2,1)$ quantile $\sqrt{u}$: the Q-Q trace is nearly affine, with an ordinary-least-squares fit $\widehat Q^{\text{emp}}(u) \approx 0.952\,\sqrt{u} - 0.155$ achieving $R^{2} = 0.996$ across the unit interval. The slope close to $1$ and the negative intercept reproduce the picture above: the empirical distribution at our deepest depth has the right $\Beta(2,1)$ {shape} and is shifted down by a finite-$N$ scale prefactor.

\begin{remark}
\label{rmk:hipster}
Two universality classes are now established for critical recursions on binary trees with a max-or-sum character at log scale. The asymmetric class, with rate $\sqrt{N}$ and limit $\Beta(2,1)$ on $[0,1]$, is exemplified by the result of Auffinger and Cable~\cite{AuffingerCable2017} for the min-plus root value and for the critical series-parallel distance (a special case of the same theorem), and by the lazy hipster random walk of Addario-Berry~et~al.~\cite{AddarioBerryCairns2020}. The symmetric class, with rate $N^{1/3}$ and limit density $\tfrac{3}{4}(1-x^{2})$ on $(-1, 1)$, is exemplified by the symmetric hipster random walk of~\cite{AddarioBerryCairns2020} and, very recently, by the critical series-parallel effective resistance, where the conjectures of Hambly and Jordan~\cite{HamblyJordan2004}, Addario-Berry~et~al.~\cite{AddarioBerryCairns2020}, and Derrida have been simultaneously resolved by Chen, Duquesne, and Shi~\cite{ChenDuquesneShi2025} and independently by Morfe~\cite{Morfe2025} within a unified framework of random homogeneous systems on a function space, parametrized by $(\alpha_{+}, \alpha_{-}, p)$. In their classification the min-plus process lies in the asymmetric class, with $\alpha_{+} > 0$ from the `$+$' branch and $\alpha_{-} = 0$ from the `$\min$' branch, consistent with the $\Beta(2,1)$ limit at rate $\sqrt{N}$.
\end{remark}

The asymptotic $\PP(\XN=1) \sim 2/N$ from Lemma~\ref{lem:px1}\,(ii) is consistent with our data within Monte Carlo error. The exact recursion (\ref{eq:y-rec}) reproduces the measured probabilities to within $1.5$ standard errors at every depth, with the largest discrepancies at the smallest accessible $N$ (Table~\ref{tab:critical_summary}, last two columns), while the leading-order $2/N$ asymptote is approached slowly from below as $N$ grows, in line with the well-known slow convergence of the conditional survival probability for critical Galton--Watson processes~\cite{AthreyaNey1972}. More precisely, carrying the substitution $z_{n} = 2/y_{n}$ in the proof of Lemma~\ref{lem:px1}\,(ii) to next order gives $z_{n+1} = z_{n} + 1 + z_{n}^{-1} + O(z_{n}^{-2})$, whence
\begin{equation}
y_{n} = 2/n - 2\log{n}/n^{2} + O(1/n^{2}),
\end{equation}
so the deviation $2/N - y_{N}$ is dominated by the logarithmic factor at the depths in Table~\ref{tab:critical_summary}.

The right panel of Figure~\ref{fig:dist_critical} displays the escape of mass behind these rates: read at fixed $k$, the empirical $\widehat{\PP}(\XN = k)$ vanishes as $N$ grows for every $k$, so no fixed integer retains probability mass in the limit and the law of $\XN$ admits no non-trivial limit on fixed values---the meaningful limit is that of the rescaled logarithm (\ref{eq:auffinger}). For $k = 1$ this is Lemma~\ref{lem:px1}\,(ii), the points approaching the reference line $2/N$ from below with $N\mku\widehat{\PP}(\XN{=}1)/2$ rising from $0.79$ at $N = 20$ to $0.90$ at $N = 60$, in line with the logarithmic correction above; for $k = 2, \ldots, 5$ the same $\Theta(1/N)$ decay is observed empirically, with $N\mku\widehat{\PP}(\XN{=}k) \approx 1.13$, $0.83$, $0.68$, $0.53$ at $N = 60$. A linear extrapolation of $\widehat{\PP}(\XN = k)$ in $1/N$, taken literally, returns intercepts of at most a few times $10^{-3}$, attributable to the logarithmic curvature rather than to surviving mass.

\begin{figure}[t]
\centering
\includegraphics[width=\linewidth]{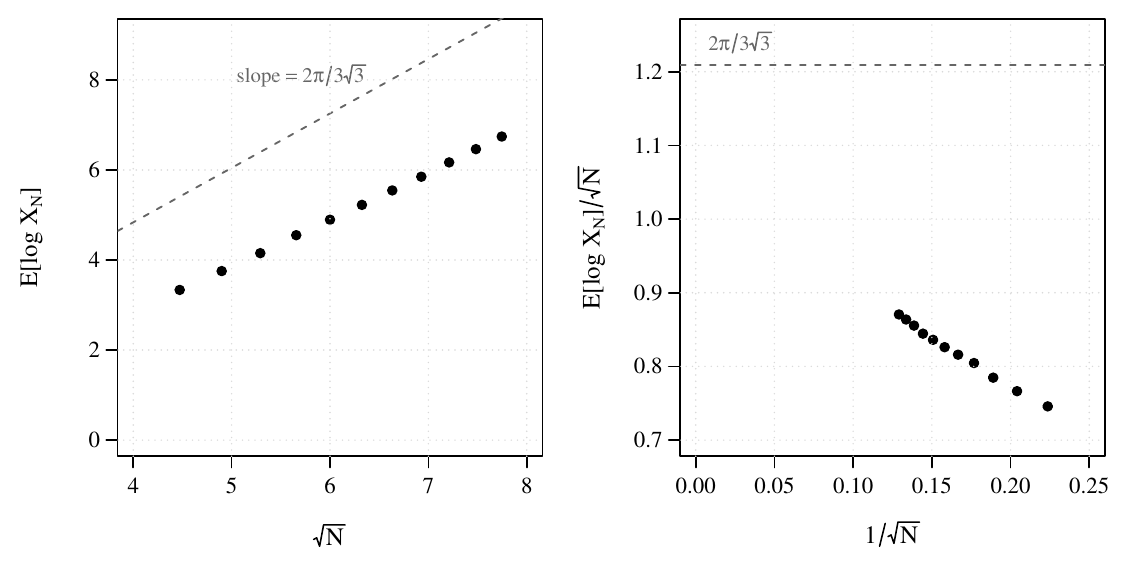}
\caption{\label{fig:critical_scaling}Test of the asymptotic of \cite[Corollary~1]{AuffingerCable2017} on our $M=10^{5}$ data at $p=1/2$. \emph{Left:} $\widehat{\EE}[\log{\XN}]$ versus $\sqrt{N}$, with the dashed reference line of theoretical slope $2\pi/3\sqrt{3} \approx 1.209$. \emph{Right:} $\widehat{\EE}[\log{\XN}]/\sqrt{N}$ versus $1/\sqrt{N}$, showing the slow approach to the theoretical limit. Error bars (1 SE) are smaller than the markers in both panels.}
\end{figure}

\begin{figure}[t]
\centering
\includegraphics[width=0.55\linewidth]{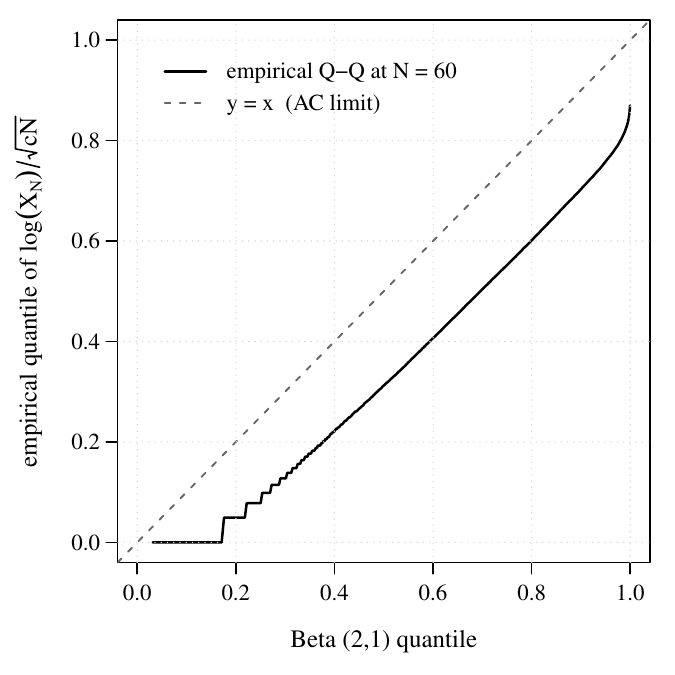}
\caption{\label{fig:qq_beta21}Quantile-quantile plot of the empirical rescaled distribution $\log{\XN}/\sqrt{cN}$ at $p=1/2$, $N=60$, $M=10^{5}$ samples (solid), against the $\Beta(2,1)$ asymptote of Theorem~1 in~\cite{AuffingerCable2017} (dashed $y = x$). The empirical Q-Q line is nearly affine; cf.~Figure~\ref{fig:critical_scaling}.}
\end{figure}

\begin{table}[t]
\caption{\label{tab:critical_summary}Summary statistics of $\XN$ for the critical min-plus process, $M=10^{5}$ samples per depth. The last two columns compare $\widehat{\PP}(\XN{=}1)$ with the recursion of Lemma~\ref{lem:px1}\,(ii) from $y_{0} = 1$; the Monte Carlo standard error is $\sigma \leq 9 \times 10^{-4}$.}
\centering
\setlength{\tabcolsep}{4pt}
\begin{tabular}{rrrrrcc}
\toprule
$N$ & $\EE[\XN]$ & $\mathrm{Var}(\XN)$ & $\max\XN$ & $\widehat{\PP}(\XN{=}1)$ & $y_{N}$ \\
\midrule
$20$ &     $93.4$ &  $1.79 \times 10^{4}$ &   $1{,}376$ & $0.0791$ & $0.0801$ \\
$24$ &    $171.4$ &  $6.95 \times 10^{4}$ &   $3{,}119$ & $0.0699$ & $0.0687$ \\
$28$ &    $304.7$ &  $2.44 \times 10^{5}$ &   $5{,}445$ & $0.0609$ & $0.0601$ \\
$32$ &    $526.7$ &  $7.83 \times 10^{5}$ &   $9{,}504$ & $0.0525$ & $0.0535$ \\
$36$ &    $866.4$ &  $2.33 \times 10^{6}$ &  $20{,}182$ & $0.0474$ & $0.0482$ \\
$40$ &  $1{,}418$ &  $6.70 \times 10^{6}$ &  $33{,}783$ & $0.0447$ & $0.0439$ \\
$44$ &  $2{,}271$ &  $1.83 \times 10^{7}$ &  $58{,}043$ & $0.0406$ & $0.0403$ \\
$48$ &  $3{,}537$ &  $4.66 \times 10^{7}$ &  $80{,}732$ & $0.0376$ & $0.0372$ \\
$52$ &  $5{,}506$ &  $1.18 \times 10^{8}$ & $150{,}243$ & $0.0349$ & $0.0346$ \\
$56$ &  $8{,}449$ &  $2.90 \times 10^{8}$ & $225{,}866$ & $0.0319$ & $0.0323$ \\
$60$ & $12{,}698$ &  $6.80 \times 10^{8}$ & $409{,}488$ & $0.0301$ & $0.0303$ \\
\bottomrule
\end{tabular}
\end{table}


\section{The off-critical min-plus process}
\label{sec:noncritical}

Auffinger and Cable \cite[Theorem~2]{AuffingerCable2017} prove that if $p<1/2$ the sequence $\XN$ converges in distribution to a non-trivial random variable $X_{\infty}(p)$, with $(\XN)_{N \geq 1}$ tight. As a special case of Lemma~\ref{lem:px1}\,(i) above (or by \cite[Lemma~16]{AuffingerCable2017}), one has the exact closed form
\begin{equation}
\label{eq:auffinger_p1}
\PP(X_{\infty}(p) = 1) = 1 - \frac{p}{1-p} = \frac{1-2p}{1-p}, \quad p < 1/2.
\end{equation}
The case $p > 1/2$ is qualitatively different: $\EE[\XN] \geq (2p)^{N}$ grows exponentially in $N$, and Lemma~\ref{lem:px1}\,(iii) shows that the mass at $1$ is wiped out at the geometric rate $\bigl(2(1-p)\bigr)^{N}$. The natural variable is then a rescaled $\XN/c_{N}$ for some norming sequence $c_{N}$, which we do not pursue here.

\subsection{The sub-critical regime}
\label{sec:subcritical}

We simulated the sub-critical min-plus process for $p = 0.10, 0.20, 0.30, 0.40$ on trees of effective depth $N=48$, with $M=10^{5}$ Monte Carlo samples per parameter, using the precomputed-leaves machinery of Appendix~\ref{sec:methods} with $K=35$ (so the simulator runs at depth $N - K = 13$). We do not push to $N=60$ because in the sub-critical regime Lemma~\ref{lem:px1}\,(i) gives geometric convergence to the closed form $y^{*} = (1-2p)/(1-p)$ at rate $(2p)^{N}$, so by $N=48$ we are already within $\sim (2p)^{48}$ of the limit---some $10^{-5}$ for $p = 0.4$ and exponentially less for smaller $p$, well below our Monte Carlo resolution. The empirical distributions $\hat{f}(k \mid \XN, p)$ are shown in Figure~\ref{fig:noncritical}, with summary statistics in Table~\ref{tab:noncritical}.

The data are consistent with the tightness established in~\cite{AuffingerCable2017}. The observed support of $\hat{f}$ shrinks rapidly as $p$ decreases, with the mean and variance of $\XN$ both converging to small finite limits well within our simulated depths (Figure~\ref{fig:noncritical}). The empirical $\widehat{\PP}(\XN=1)$ values agree with the closed form (\ref{eq:auffinger_p1}) to within one Monte Carlo standard error ($\sigma \leq 1.6 \times 10^{-3}$) at every $p$ on our grid (Table~\ref{tab:noncritical}, second and third columns). The mean $\widehat{\EE}[\XN]$ rises by nearly an order of magnitude across the sub-critical interval, from $\approx 1.14$ at $p = 0.1$ to $\approx 7.84$ at $p = 0.4$; combined with the geometric decay of $\widehat{\PP}(\XN \geq k)$ in Figure~\ref{fig:noncritical}, this is consistent with $X_{\infty}(p)$ being approximately geometric with tail ratio approaching $1$ as $p \nearrow 1/2$. A thorough characterization of $X_{\infty}(p) \mid X_{\infty}(p) > 1$ across the sub-critical range is beyond our scope.

\begin{figure}[t]
\centering
\includegraphics[width=0.6\linewidth]{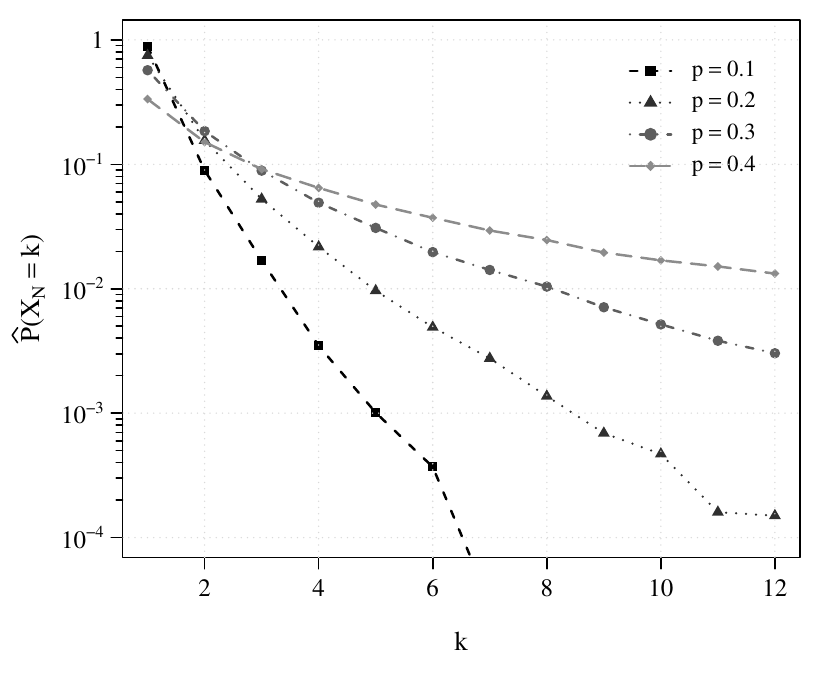}
\caption{\label{fig:noncritical}Empirical distribution $\hat{f}(k \mid \XN, p)$ of the root value of the sub-critical min-plus process for $p = 0.10$, $0.20$, $0.30$, $0.40$, at effective tree depth $N=48$ and $M=10^{5}$ samples per parameter. Vertical axis on a logarithmic scale.}
\end{figure}

\begin{table}[t]
\caption{\label{tab:noncritical}Moments of $\XN$ for the sub-critical min-plus process at $N=48$, $M=10^{5}$ samples per $p$. The second and third columns give $\widehat{\PP}(\XN{=}1)$ and the closed form $(1-2p)/(1-p)$ from Lemma~\ref{lem:px1}\,(i)/eq.~(\ref{eq:auffinger_p1}).}
\centering
\setlength{\tabcolsep}{4pt}
\begin{tabular}{cccccc}
\toprule
$p$ & $\widehat{\PP}(\XN{=}1)$ & $(1{-}2p)/(1{-}p)$ & $\widehat{\EE}[\XN]$ & $\widehat{\mathrm{Var}}(\XN)$ & $\max\XN$ \\
\midrule
$0.10$ & $0.8886$ & $0.8889$ & $1.140$ & $0.199$ & $11$  \\
$0.20$ & $0.7498$ & $0.7500$ & $1.430$ & $0.966$ & $20$  \\
$0.30$ & $0.5708$ & $0.5714$ & $2.261$ & $6.55$  & $53$  \\
$0.40$ & $0.3344$ & $0.3333$ & $7.839$ &   $251$ & $511$ \\
\bottomrule
\end{tabular}
\end{table}


\subsection{The super-critical regime}
\label{sec:supercritical}

For $p > 1/2$ Lemma~\ref{lem:px1}\,(iii) implies two simultaneous behaviors: the mass at $1$ decays geometrically, $\PP(\XN=1) = O((2(1-p))^{N})$, while the mean grows at least exponentially, $\EE[\XN] \geq (2p)^{N}$. To document the trichotomy of Lemma~\ref{lem:px1} on its third regime, we ran a complementary sweep at $p = 0.6,\, 0.7$ and effective depths $N = 32,\, 40,\, 48$, $M=10^{5}$ Monte Carlo samples per $(p, N)$, using the precomputed-leaves machinery of Appendix~\ref{sec:methods} with $K=20$ for $N=32$ and $K=35$ for $N=40, 48$. Summary statistics are collected in Table~\ref{tab:supercritical}.

\begin{table}[t]
\caption{\label{tab:supercritical}Summary statistics of $\XN$ for the super-critical min-plus process at $p = 0.6,\, 0.7$ and $N = 32,\, 40,\, 48$, $M=10^{5}$ samples per $(p, N)$. The third column $(2p)^{N}$ is the lower bound on $\EE[\XN]$ from Lemma~\ref{lem:px1}\,(iii). The theoretical geometric decay of $\widehat{\PP}(\XN{=}1)$ is $(2(1-p))^{N}$: resolvable on the MC scale for $p=0.6$ and well below resolution for $p=0.7$.}
\centering
\setlength{\tabcolsep}{4pt}
\begin{tabular}{cccccc}
\toprule
$p$ & $N$ & $(2p)^{N}$ & $\widehat{\EE}[\XN]$ & $\widehat{\mathrm{Var}}(\XN)$ & $\widehat{\PP}(\XN{=}1)$ \\
\midrule
$0.6$ & $32$ & $3.4 \times 10^{2}$ & $8.63 \times 10^{4}$ & $6.31 \times 10^{9}$  & $1.5 \times 10^{-4}$ \\
$0.6$ & $40$ & $1.5 \times 10^{3}$ & $1.29 \times 10^{6}$ & $1.42 \times 10^{12}$ & $4.0 \times 10^{-5}$ \\
$0.6$ & $48$ & $6.3 \times 10^{3}$ & $1.91 \times 10^{7}$ & $3.10 \times 10^{14}$ & $< 10^{-5}$ \\
$0.7$ & $32$ & $4.7 \times 10^{4}$ & $3.62 \times 10^{6}$ & $5.19 \times 10^{12}$ & $< 10^{-5}$ \\
$0.7$ & $40$ & $7.0 \times 10^{5}$ & $2.33 \times 10^{7}$ & $2.03 \times 10^{14}$ & $< 10^{-5}$ \\
$0.7$ & $48$ & $1.0 \times 10^{7}$ & $9.75 \times 10^{8}$ & $3.71 \times 10^{17}$ & $< 10^{-5}$ \\
\bottomrule
\end{tabular}
\end{table}

The empirical growth exponent $\widehat{r}_{N}(p) \deq \log{\widehat{\EE}[\XN]}/N$ from Table~\ref{tab:supercritical} is well above the lemma's lower-bound exponent $\log{(2p)}$ at every depth in our grid: for $p=0.6$, $\widehat{r}_{N} = 0.355$, $0.352$, $0.349$ at $N = 32, 40, 48$ (versus $\log{(2p)} = 0.182$); for $p=0.7$, $\widehat{r}_{N} = 0.472$, $0.424$, $0.431$ (versus $\log{(2p)} = 0.337$). The empirical ratio $\widehat{\EE}[\XN]/(2p)^{N}$ increases with $N$ (from $252$ at $N=32$ to $3{,}018$ at $N=48$ for $p=0.6$), so the lower bound $\EE[\XN] \geq (2p)^{N}$ of Lemma~\ref{lem:px1}\,(iii)---which considers only the all-`$+$' realization---is not asymptotically tight on this grid; the empirical exponent settles near $\sim 0.35$ for $p=0.6$ and $\sim 0.43$ for $p=0.7$ at the depths simulated.

For $p=0.6$, the empirical $\widehat{\PP}(\XN=1)$ values $1.5 \times 10^{-4}$, $4.0 \times 10^{-5}$, and $0$ at $N = 32, 40, 48$ are at roughly $0.19$, $0.30$, and $0$ times the theoretical geometric values $0.8^{N} \approx 7.9 \times 10^{-4}$, $1.3 \times 10^{-4}$, $2.2 \times 10^{-5}$, consistent in order of magnitude with the $O\bigl((2(1-p))^{N}\bigr)$ decay of Lemma~\ref{lem:px1}\,(iii) with a constant below $1$; the $N=48$ count of $0$ out of $M=10^{5}$ samples is consistent with the expected count of $2.2$ (Poisson zero-probability $e^{-2.2} \approx 0.11$). For $p=0.7$, the theoretical decay $0.6^{N}$ is already $\sim 8 \times 10^{-8}$ at $N=32$ and smaller for larger $N$, well below MC resolution; the empirical $\widehat{\PP}(\XN=1) = 0$ at all three depths is consistent with the faster theoretical decay but does not estimate the rate.


\subsection{The near-critical scaling window}
\label{sec:nearcritical}

The fixed-$p$ dichotomy above---tightness for $p < 1/2$ and exponential growth for $p > 1/2$---motivates a near-critical experiment. For each effective depth $N$, we fix $p = p_{N} = 1/2 - c/N$ throughout the tree. On the super-critical side, the analogous series-parallel distance has growth exponent $\alpha(1/2+\eps) \asymp \sqrt{\eps}$, eq.~(\ref{eq:cdds}). Matching $N\sqrt{\eps}$ to the critical $\sqrt{N}$ scale selects $\eps \asymp N^{-1}$, motivating the sub-critical window considered here.

We simulated the process with all-ones leaves at $c  = 1/2, 1, 2$ and $N = 24, 32, 40$, $48, 56, 60$, and $M = 10^{5}$ samples per cell, using the machinery of Appendix~\ref{sec:methods} with a separate leaf table per $(c,N)$ (the level recursion depends on $p$; $K = 20$ for $N \leq 32$, $K = 35$ otherwise) and the critical runs of Section~\ref{sec:critical} as the $c = 0$ baseline. The $c$-doubling pairs $(1/2, 24)$--$(1, 48)$ and $(1, 24)$--$(2, 48)$ visit the same $p_{N}$ ($23/48$ and $11/24$, respectively) at two different depths, and so discriminate whether the finite-depth behavior depends on $p_{N}$ alone or on the window parameter $c = N(1/2 - p_{N})$.

Table~\ref{tab:nearcritical} and Figure~\ref{fig:nearcritical} collect the rescaled means $\widehat{A}(c, N) \deq \widehat{\EE}[\log{\XN}]/\sqrt{N}$. Across the sampled depths, each $c$-column follows the critical column nearly in parallel. The deficit $\widehat{\varphi}(c, N) \deq \widehat{A}(0, N) - \widehat{A}(c, N)$ has column averages $0.159$, $0.291$, and $0.471$ for $c = 1/2$, $1$, and $2$, respectively, and varies by at most about $5\%$. Thus the data are consistent with $\EE[\log{\XN}] \approx \EE[\log{\XN}]_{\mathrm{crit}} - \varphi(c)\sqrt{N}$. The equal-$p_{N}$ pairs instead have markedly different deficits, showing that $p_{N}$ alone does not determine the finite-depth behavior and supporting $c = N(1/2 - p_{N})$ as the relevant scaling variable.

\begin{figure}[t]
\centering
\includegraphics[width=\linewidth]{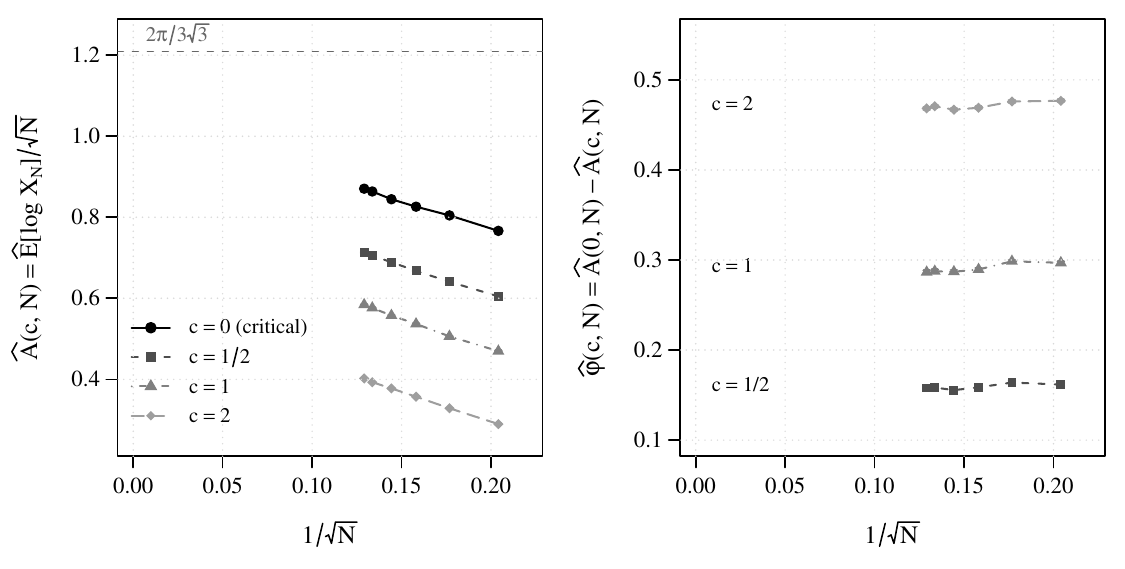}
\caption{\label{fig:nearcritical}The near-critical window $p_{N} = 1/2 - c/N$ at $c \in \{1/2, 1, 2\}$, effective depths $N = 24,\, 32,\, 40,\, 48,\, 56,\, 60$, $M = 10^{5}$ samples per cell, with the critical runs of Section~\ref{sec:critical} as the $c = 0$ baseline. \emph{Left:} rescaled mean $\widehat{A}(c, N) = \widehat{\EE}[\log{\XN}]/\sqrt{N}$ versus $1/\sqrt{N}$; the dashed line marks the critical asymptote $2\pi/3\sqrt{3}$ of (\ref{eq:AClim}). \emph{Right:} deficit $\widehat{\varphi}(c, N) = \widehat{A}(0, N) - \widehat{A}(c, N)$, nearly $N$-independent over the sampled grid, with a total spread of at most about $5\%$ of its column average. Error bars ($1$ SE $\leq 2 \times 10^{-3}$) are smaller than the markers in both panels.}
\end{figure}

\begin{table}[t]
\caption{\label{tab:nearcritical}Rescaled mean $\widehat{A}(c, N) = \widehat{\EE}[\log{\XN}]/\sqrt{N}$ for the near-critical min-plus process at $p_{N} = 1/2 - c/N$, $M = 10^{5}$ samples per cell; the $c = 0$ column is the critical baseline of Section~\ref{sec:critical}. The Monte Carlo standard error is $\leq 1.3 \times 10^{-3}$ in every cell. Note the equal-$p_{N}$ pairs $(c, N) = (1/2, 24)$--$(1, 48)$ and $(1, 24)$--$(2, 48)$.}
\centering
\setlength{\tabcolsep}{4pt}
\begin{tabular}{ccccc}
\toprule
$N$ & $c=0$ & $c=1/2$ & $c=1$ & $c=2$ \\
\midrule
$24$ & $0.7663$ & $0.6047$ & $0.4695$ & $0.2895$ \\
$32$ & $0.8045$ & $0.6407$ & $0.5060$ & $0.3285$ \\
$40$ & $0.8261$ & $0.6675$ & $0.5366$ & $0.3570$ \\
$48$ & $0.8444$ & $0.6890$ & $0.5575$ & $0.3775$ \\
$56$ & $0.8636$ & $0.7053$ & $0.5762$ & $0.3928$ \\
$60$ & $0.8705$ & $0.7127$ & $0.5841$ & $0.4021$ \\
\bottomrule
\end{tabular}
\end{table}

The atom at $1$ admits an analytic description in this window. With $p = p_{N}$ held fixed throughout a depth-$N$ tree, a standard Euler-scheme scaling of the exact recursion (\ref{eq:y-rec}) gives
\[
N\PP(\XN = 1) \longrightarrow \frac{4c}{1-e^{-2c}}.
\]
Indeed, writing $y_{\lfloor Nt\rfloor} \sim u(t)/N$ yields $u' = 2cu-u^{2}/2$ with entrance condition $u(0+) = \infty$, and hence $u(t) = 4c/(1-e^{-2ct})$. The coefficient tends to $2$ as $c \downarrow 0$, recovering the critical asymptotic of Lemma~\ref{lem:px1}\,(ii), and is asymptotic to $4c$ as $c \to \infty$, matching the fixed-$p$ equilibrium $y^{*} \sim 4c/N$. The empirical atoms at $N = 60$ are consistent with the exact recursion: $0.1256$ versus $0.1265$ at $c = 2$, and $0.0498$ versus $0.0491$ at $c = 1/2$.

For the logarithmic mean, the grid does not settle the functional form of $\varphi$: doubling $c$ multiplies the deficit by $\approx 1.83$ and then $\approx 1.62$, between the $\sqrt{c}$ of the super-critical heuristic and linear growth. Whether $\widehat{A}(c, N) \to 2\pi/3\sqrt{3} - \varphi_{\infty}(c)$ as $N \to \infty$, and how $\varphi_{\infty}$ relates to (\ref{eq:cdds}), are open questions to which we return in Section~\ref{sec:summary}.


\section{The min-plus process with Bernoulli leaves}
\label{sec:variant}

We can extend the min-plus process by introducing Bernoulli initial conditions. In this variant, the initial values $X_{0}(i)$ become \iid\ distributed as $\PP(X_{0}(i)=0)=q$ and $\PP(X_{0}(i)=1)=1-q$, with $0 \leq q \leq 1$. When $p=1/2$ and $q=0$ we recover the critical min-plus process of Section~\ref{sec:critical}.

As $q \to 1$ the leaves become almost surely $0$, so $\XN \to 0$ deterministically. The behavior for $q$ in between is governed by the following lemma, which fully determines the order parameter $\PP(\XN=0)$ at $p=1/2$ and locates the absorbing-state phase transition in $p$ rather than in $q$.

\begin{lemma}
\label{lem:px0}
Let $X_{0}(i)$ be \iid\ nonnegative integer random variables with $\PP(X_{0}(i)=0) = q \in [0,1]$, and let $x_{n} \deq \PP(X_{n}=0)$. Then $x_{n}$ obeys the deterministic recursion
\begin{equation}
\label{eq:bern-rec}
x_{n+1} = 2(1-p)\,x_{n} + (2p-1)\,x_{n}^{2}, \quad x_{0} = q.
\end{equation}
For $p \neq 1/2$, the only fixed points of $x \mapsto 2(1-p)x + (2p-1)x^{2}$ in $[0,1]$ are $x = 0$ and $x = 1$. At $p=1/2$ the map is the identity, so every $x \in [0,1]$ is fixed. The dynamics in the three regimes are:
\begin{enumerate}
\item[(i)] At $p=1/2$ the recursion is the identity, so $\PP(\XN=0) = q$ for every $N \geq 0$.
\item[(ii)] For $0 < p < 1/2$ the fixed point $x = 1$ is locally stable with $1 - x_{n+1} \sim 2p(1 - x_{n})$, so for any $q > 0$, $\PP(\XN=0) \to 1$ at the geometric rate $1 - \PP(\XN=0) = O((2p)^{N})$.
\item[(iii)] For $1/2 < p < 1$ the fixed point $x = 0$ is locally stable with $x_{n+1} \sim 2(1-p)\,x_{n}$, so for any $q < 1$, $\PP(\XN=0) \to 0$ at the geometric rate $\PP(\XN=0) = O((2(1-p))^{N})$.
\end{enumerate}
\end{lemma}
\begin{proof}
Conditioning on the operator at the parent of the nodes $X_{n}^{(1)}, X_{n}^{(2)}$,
\[
\PP(X_{n+1} = 0) = p\PP(X_{n}^{(1)}+X_{n}^{(2)} = 0) + (1-p)\PP(\min(X_{n}^{(1)}, X_{n}^{(2)}) = 0).
\]
Since the $X_{n}$ are nonnegative integers, $X_{n}^{(1)} + X_{n}^{(2)} = 0$ iff both summands vanish, with probability $x_{n}^{2}$, while $\min(X_{n}^{(1)}, X_{n}^{(2)}) = 0$ iff at least one of them does, with probability $1 - (1-x_{n})^{2} = 2x_{n} - x_{n}^{2}$. Hence
\[
x_{n+1} = p\,x_{n}^{2} + (1-p)(2x_{n} - x_{n}^{2}) = 2(1-p)\,x_{n} + (2p-1)\,x_{n}^{2},
\]
which is (\ref{eq:bern-rec}). The fixed-point equation factorizes as $(2p-1)\,x(1-x)=0$, so for $p \neq 1/2$ the only fixed points are $0$ and $1$, while at $p=1/2$ every $x \in [0,1]$ is fixed. The stated geometric rates follow from $f'(0) = 2(1-p)$ and $f'(1) = 2p$ in the linearizations near each fixed point. At the endpoints, $1-x_{N} = (1-q)^{2^{N}}$ for $p=0$ and $x_{N} = q^{2^{N}}$ for $p=1$; the convergence is thus supergeometric.
\end{proof}

\begin{remark}
\label{rmk:duality}
Lemmas~\ref{lem:px1} and~\ref{lem:px0} are dual one-dimensional recursions for the marginals $\PP(X_{n} = 1)$ and $\PP(X_{n} = 0)$, respectively, both obtained by the same conditioning argument. They share the linearization slope $f'(0) = 2(1-p)$ at the lower fixed point and consequently exhibit the same change of stability at $p_{c} = 1/2$, which is the structural origin of the absorbing-state phase transition in this section. The two lemmas differ in their upper structure because $0$ is the additive identity but $1$ is not. In Lemma~\ref{lem:px0} both the `$\min$' and the `$+$' branches preserve the event $\{X_{n} = 0\}$, and at $p=1/2$ their convex combination collapses to the identity, so $\PP(X_{n} = 0) = q$ is exactly conserved. In Lemma~\ref{lem:px1} only the `$\min$' branch preserves $\{X_{n} = 1\}$ while the `$+$' branch acts as a sink, so the recursion at $p=1/2$ contracts polynomially toward $0$ as $2/n$ rather than freezing.
\end{remark}

Since $q$ plays no role beyond setting the limit value along the critical line $p = 1/2$, our sweeps in this section traverse it rather than crossing a transition in $q$; the immediate empirical questions are how cleanly the data lie on the exact line $\PP(\XN=0) = q$ and what the conditional law $\XN \mid \XN > 0$ looks like along it.

To verify the identity of Lemma~\ref{lem:px0}\,(i) we ran a sweep of $M=10^{5}$ Monte Carlo samples at effective depth $N=48$ over $q = 0$, $0.05$, $0.10$, \ldots, $0.95$, $1.00$, using the precomputed-leaves machinery of Appendix~\ref{sec:methods} with $K=35$ (one $\mu_{K}$ table per $q$, simulator at depth $N-K = 13$); see Figure~\ref{fig:bern_p0}. The empirical points lie on the diagonal $\PP(\XN=0) = q$ across the full unit interval, with a maximum residual of $3.6 \times 10^{-3}$ consistent with the $\sqrt{q(1-q)/M} \leq 1.6 \times 10^{-3}$ Monte Carlo standard error per point at $M=10^{5}$.

While the order parameter is determined exactly by Lemma~\ref{lem:px0}, the conditional distribution of $\XN$ given $\XN > 0$ at $p=1/2$ remains non-trivial. Figure~\ref{fig:bern_moments} and Table~\ref{tab:bern_n48} report the moments of $\XN$ as a function of $q$. The mean drops by more than four orders of magnitude across the interval $[0,\,1]$, from $\approx 3{,}570$ at $q=0$ to $\approx 0.12$ at $q=0.95$ and exactly $0$ at $q=1$, even though the order parameter $\PP(\XN=0)$ moves only linearly from $0$ to $1$. The collapse is concentrated in the conditional law itself: the rescaled mean $\widehat{\EE}[\log{\XN} \mid \XN > 0]/\sqrt{cN}$ at $N=48$ falls smoothly and monotonically from the all-ones critical value $0.466$ at $q=0$ to $0.290$ at $q=1/2$ and $0.050$ at $q=0.95$. The empirical profile lies between the naive $(1-q)$- and $\sqrt{1-q}$-rescalings of the $q=0$ profile and matches neither.

To check that the observed $q$-dependence of the conditional law is not merely an artifact of the single depth $N=48$, we ran a complementary sweep at $q = 0.25, 0.50, 0.75$ and effective depths $N = 36, 48, 60$ (using $K = 35$ leaf tables, $M = 10^{5}$ samples per $(q, N)$). The conditional rescaled means $\widehat\EE[\log{\XN} \mid \XN > 0]/\sqrt{cN}$ collected in Table~\ref{tab:bern_nstability} preserve the $q$-ordering and inter-curve spacing essentially unchanged across the three depths: the spread between $q=0.25$ and $q=0.75$ is $\approx 0.20$ at all three $N$, while each row drifts upward with $N$ by $\sim 0.02$--$0.03$ over $N = 36 \to 60$, comparable to the corresponding drift of the all-ones critical case at the same depths.

\begin{table}[t]
\caption{\label{tab:bern_nstability}Conditional rescaled mean $\widehat\EE[\log{\XN} \mid \XN > 0]/\sqrt{cN}$ for the Bernoulli variant at $p = 1/2$, $K = 35$, $M = 10^{5}$ samples per $(q, N)$.}
\centering
\setlength{\tabcolsep}{8pt}
\begin{tabular}{cccc}
\toprule
$q$ & $N = 36$ & $N = 48$ & $N = 60$ \\
\midrule
$0.25$ & $0.367$ & $0.384$ & $0.396$ \\
$0.50$ & $0.275$ & $0.290$ & $0.301$ \\
$0.75$ & $0.166$ & $0.177$ & $0.184$ \\
\bottomrule
\end{tabular}
\end{table}

The recursion (\ref{eq:bern-rec}) extends without change to arbitrary nonnegative-integer initial distributions $X_{0}(i)$ with $\PP(X_{0}(i)=0) = q$, since the proof depends only on the marginal mass at $0$. So at $p=1/2$ the order parameter $\PP(\XN=0)$ is universal in the leaf distribution and depends only on $q$. The corresponding question for the conditional law $\XN \mid \XN > 0$---whether it inherits the same leaf universality---is sharpened by the random-homogeneous-systems framework of~\cite{ChenDuquesneShi2025,Morfe2025}: their classifying parameters $(\alpha_{+},\alpha_{-},p) = (\alpha_{+},0,1/2)$ depend on the operator distribution but not on the leaf marginal, and their Conjecture~1.12(b) predicts that the asymmetric $\sqrt{N}$ regime with $\Beta(2,1)$ limit should persist for any nonnegative-integer leaves with $q < 1$. The finite-$N$ deformation visible in Table~\ref{tab:bern_n48} and Figure~\ref{fig:bern_moments} is then either slow convergence to that limit at depths beyond our reach---consistent with the conjecture---or evidence that the asymmetric class is finer than the operator-marginal classification implies. Whether heavier-tailed leaves split the class into a tail-dependent family of limit laws analogous to the $\chi(\alpha) = 1/(\alpha+2)$ exponents of~\cite{HuShi2018} in the Derrida--Retaux setting is a related but distinct question that we leave to future work.

\begin{figure}[t]
\centering
\includegraphics[width=0.6\linewidth]{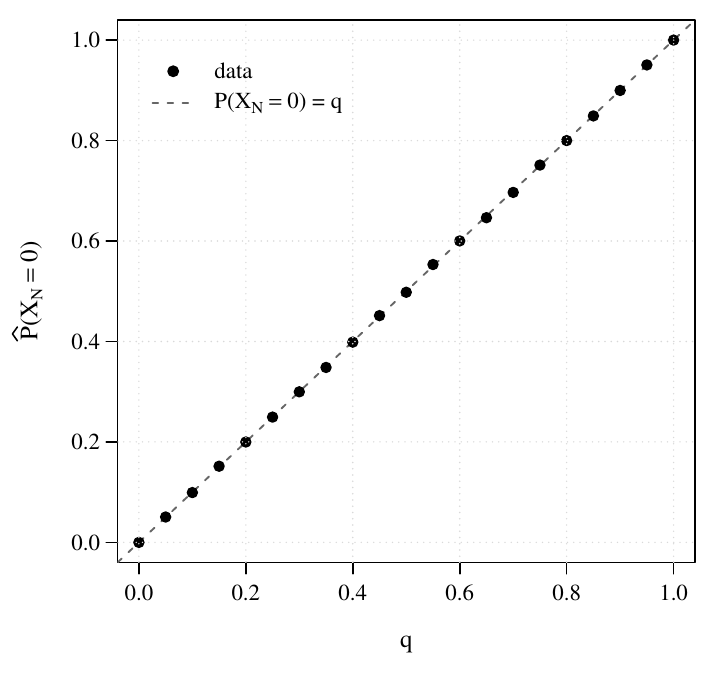}
\caption{\label{fig:bern_p0}Empirical confirmation of the identity $\PP(\XN=0) = q$ from Lemma~\ref{lem:px0}~(i): the order parameter $\widehat{\PP}(\XN=0)$ for the Bernoulli variant at $p=1/2$, effective depth $N=48$, $M=10^{5}$ samples per $q$, plotted against the leaf-zero probability $q$. The dashed diagonal is the exact line. Empirical points lie on the diagonal across $[0,1]$ to within Monte Carlo error.}
\end{figure}

\begin{figure}[t]
\centering
\includegraphics[width=0.6\linewidth]{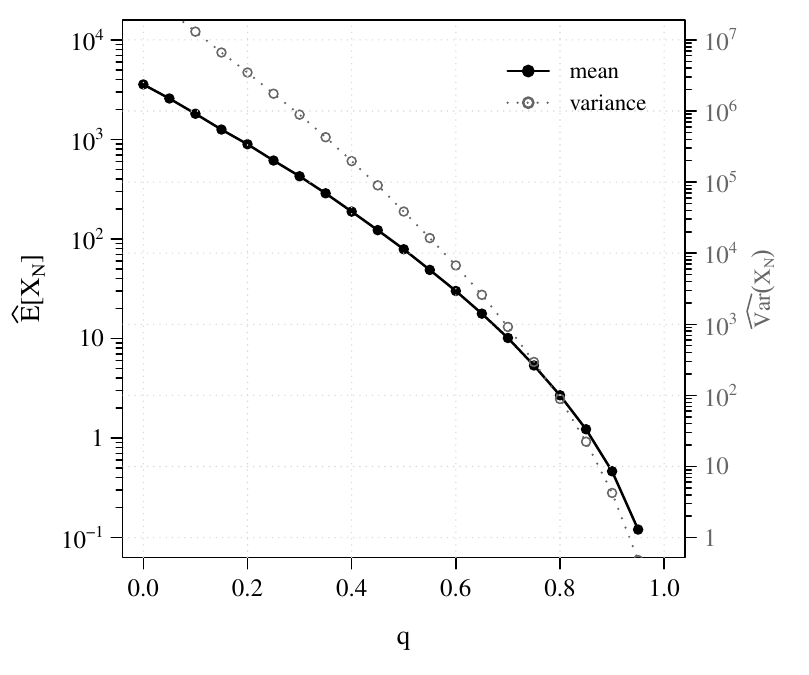}
\caption{\label{fig:bern_moments}Mean (left axis) and variance (right axis) of $\XN$ for the Bernoulli variant at $p=1/2$ and $N=48$, as a function of $q$. Both quantities decay by several orders of magnitude across the unit interval and vanish at $q=1$ (not shown).}
\end{figure}

\begin{table}[t]
\caption{\label{tab:bern_n48}Summary statistics of the Bernoulli variant at $p=1/2$, $N=48$, $M=10^{5}$ samples per $q$. The column $\widehat{\PP}(\XN=0)$ is the empirical order parameter; Lemma~\ref{lem:px0}~(i) gives $\PP(\XN=0) = q$ exactly; the empirical values agree within Monte Carlo error ($\sigma \approx 1.6 \times 10^{-3}$).}
\centering
\setlength{\tabcolsep}{4pt}
\begin{tabular}{cccccc}
\toprule
$q$ & $\widehat{\PP}(\XN=0)$ & $\widehat{\EE}[\XN]$ & $\widehat{\mathrm{Var}}(\XN)$ & $\max\XN$ & $\widehat{\EE}[\XN \mid \XN > 0]$ \\
\midrule
$0.00$ & $0.0000$ & $3568.4$ &  $4.69 \times 10^{7}$ & $86{,}768$ & $3568.4$ \\
$0.05$ & $0.0505$ & $2578.8$ &  $2.55 \times 10^{7}$ & $61{,}403$ & $2716.1$ \\
$0.10$ & $0.0992$ & $1812.4$ &  $1.31 \times 10^{7}$ & $47{,}565$ & $2012.0$ \\
$0.15$ & $0.1517$ & $1257.4$ &  $6.65 \times 10^{6}$ & $40{,}018$ & $1482.2$ \\
$0.20$ & $0.1999$ & $895.0$  &  $3.48 \times 10^{6}$ & $25{,}801$ & $1118.6$ \\
$0.25$ & $0.2495$ & $613.8$  &  $1.75 \times 10^{6}$ & $18{,}527$ &  $817.8$ \\
$0.30$ & $0.2998$ & $426.7$  &  $8.89 \times 10^{5}$ & $13{,}108$ &  $609.4$ \\
$0.35$ & $0.3483$ & $287.5$  &  $4.28 \times 10^{5}$ & $9{,}418$  &  $441.2$ \\
$0.40$ & $0.3987$ & $188.4$  &  $1.97 \times 10^{5}$ & $6{,}279$  &  $313.4$ \\
$0.45$ & $0.4515$ & $122.7$  &  $9.00 \times 10^{4}$ & $4{,}310$  &  $223.6$ \\
$0.50$ & $0.4980$ & $78.86$  &  $3.85 \times 10^{4}$ & $2{,}501$  &  $157.1$ \\
$0.55$ & $0.5532$ & $48.89$  &  $1.63 \times 10^{4}$ & $2{,}407$  &  $109.4$ \\
$0.60$ & $0.6002$ & $30.03$  &  $6.73 \times 10^{3}$ & $1{,}151$  &  $75.12$ \\
$0.65$ & $0.6464$ & $17.77$  &  $2.60 \times 10^{3}$ & $787$      &  $50.26$ \\
$0.70$ & $0.6968$ & $10.13$  & $918.2$               &   $526$    &  $33.40$ \\
$0.75$ & $0.7512$ & $5.363$  & $295.3$               &   $337$    &  $21.55$ \\
$0.80$ & $0.7999$ & $2.688$  & $89.0$                &   $206$    &  $13.44$ \\
$0.85$ & $0.8490$ & $1.225$  & $22.3$                &   $101$    &   $8.12$ \\
$0.90$ & $0.8997$ & $0.464$  & $4.25$                &    $39$    &   $4.63$ \\
$0.95$ & $0.9506$ & $0.120$  & $0.482$               &    $17$    &   $2.44$ \\
$1.00$ & $1.0000$ & $0.000$  & $0$                   &     $0$    &    ---   \\
\bottomrule
\end{tabular}
\end{table}


\section{Summary and discussion}
\label{sec:summary}

The simulations above test and refine the rigorous picture of the min-plus process on the binary tree. At $p=1/2$ with all-ones leaves (Section~\ref{sec:critical}), we find finite-depth numerical agreement with the stretched-exponential asymptotic of Theorem~1 in~\cite{AuffingerCable2017} at effective tree depths $N$ up to $60$: the empirical Q-Q line against the theoretical $\Beta(2,1)$ profile is nearly affine, indicating that the theoretical distributional shape is reproduced at finite $N$ up to a $1/\sqrt{N}$ scale-prefactor lag that puts the empirical mean $\approx 12\%$ short of the asymptote at $N \leq 60$. In the sub-critical regime (Section~\ref{sec:subcritical}) Lemma~\ref{lem:px1}\,(i) yields the closed form $\PP(\XN=1) \to (1-2p)/(1-p)$ at geometric rate $(2p)^{N}$, and the empirical sweep matches the closed form within a Monte Carlo standard error at every sampled $p$. The super-critical sweep (Section~\ref{sec:supercritical}) at $p = 0.6,\, 0.7$ and $N = 32,\, 40,\, 48$ documents an empirical growth exponent $\log{\widehat{\EE}[\XN]}/N \approx 0.35$ ($p=0.6$) and $\approx 0.43$ ($p=0.7$), well above the lemma's all-`$+$' lower-bound exponent $\log{(2p)}$, so the lower bound of Lemma~\ref{lem:px1}\,(iii) is not asymptotically tight at the depths simulated. In the near-critical window $p_{N} = 1/2 - c/N$ (Section~\ref{sec:nearcritical}), the rescaled mean is consistent across the sampled grid with critical $\sqrt{N}$ scaling and a $c$-dependent amplitude deficit that is nearly $N$-independent. The equal-$p_{N}$ pairs show that $p_{N}$ alone does not determine the finite-depth behavior and support $c = N(1/2 - p_{N})$ as the relevant scaling variable. In the Bernoulli-leaves variant (Section~\ref{sec:variant}), Lemma~\ref{lem:px0}~(i) pins the order parameter $\PP(\XN=0) = q$ exactly at $p=1/2$ for every $N$, locating the absorbing-state transition in the operator-mixing probability $p_{c}=1/2$ rather than in $q$. The conditional law on positives is found to deform substantially with $q$ at the simulated depth, indicating that the $\Beta(2,1)$ profile of the all-ones critical case need not persist in the conditional positive law of the Bernoulli-leaf model.

Three existing approaches seem relevant to making the observations above rigorous. First, the XY-coupling technique developed for the Derrida--Retaux model in~\cite{ChenDagard2021} has the right structure for free-energy bounds on $\PP(\XN > 0)$ near the absorbing-state transition of Section~\ref{sec:variant}. Second, the recursive-tree-process and endogeny framework of~\cite{AldousBandyopadhyay2005} fits our sub-critical regime $p < 1/2$, and would in principle yield an alternative proof of the convergence-in-distribution result of~\cite{AuffingerCable2017}. Third, the PDE/viscosity-solution approach developed by Addario-Berry, Beckman, and Lin~\cite{AddarioBerryBeckmanLin2022} for cooperative motions adapts naturally to recurrences for the marginal $\mu_{n}$ of the min-plus process and has been worked out for $\Beta(2,1)$-class scaling limits.

Recently, Chen, Duquesne, and Shi~\cite{ChenDuquesneShi2025} and, independently, Morfe~\cite{Morfe2025} proved Derrida's $N^{1/3}$ conjecture in the symmetric case (cf.~Remark~\ref{rmk:hipster}), settling the Hambly--Jordan effective-resistance limit and identifying a unifying framework of random homogeneous systems parametrized by $(\alpha_{+},\alpha_{-},p)$. The min-plus process is asymmetric in their classification and remains in the $\sqrt{N}$ regime governed by~\cite{AuffingerCable2017}. For the closely related graph-distance recursion on the same series-parallel graph---which~\cite{AuffingerCable2017} covers at criticality---Chen, Derrida, Duquesne, and Shi~\cite{ChenDerridaDuquesneShi2026} have recently determined the slightly super-critical exponent
\begin{equation}
\label{eq:cdds}
\alpha(p_{c}+\eps) \deq \lim_{n} \frac{1}{n}\log{\EE[D_{n}(p_{c}+\eps)]} \sim \frac{\pi}{\sqrt{6}}\sqrt{\eps},
\end{equation}
a Berezinskii--Kosterlitz--Thouless-type result in the spirit of the Derrida--Retaux line. The asymmetric side has otherwise been less probed than the now-rigorous symmetric one; the present finite-depth numerics address it. This two-class picture frames the open question of Section~\ref{sec:variant}: whether the conditional law $\XN \mid \XN > 0$ of the Bernoulli variant stays in the asymmetric $\sqrt{N}$ class as $q \to 1$. The near-critical window of Section~\ref{sec:nearcritical} poses a complementary one: whether $\EE[\log{\XN}]$ at $p_{N} = 1/2 - c/N$ converges after $\sqrt{N}$ rescaling to a limit $2\pi/3\sqrt{3} - \varphi_{\infty}(c)$, and whether $\varphi_{\infty}$ is governed by the same $\sqrt{\eps}$ mechanism as (\ref{eq:cdds})---a scaling-window question in the spirit of the Mallows-measure analyses~\cite{MuellerStarr2013}.


\appendix

\section{Algorithm and computational methodology}
\label{sec:methods}

The simulations were carried out with a single-file C11 program ($\sim 600$ lines of code) developed for this work; data analysis and figures were done in R. All computations ran on an Apple M1~Pro chip (2021 vintage). The code---the simulator, the leaf-table precompute tool, and the run and analysis scripts---is available in a public GitHub repository and archived on Zenodo under DOI \href{https://doi.org/10.5281/zenodo.21385814}{10.5281/zenodo.21385814}. Here we describe our algorithmic choices and their CPU and statistical bottlenecks.

\subsection{The level-wise recursion and FFT-based precompute}\label{sec:appendix-recursion}

The dual lemmas of Sections~\ref{sec:critical} and~\ref{sec:variant} are the $j=1$ and $j=0$ slices of a more general one-step recursion that propagates the full marginal $\mu_{n}(j) \deq \PP(X_{n} = j)$ through one level of the tree. Conditioning on the operator at the parent of the nodes $X_{n}^{(1)}, X_{n}^{(2)}$ as in the proofs of Lemmas~\ref{lem:px1} and~\ref{lem:px0},
\begin{equation}
\label{eq:full-rec}
\mu_{n+1}(j)  = p(\mu_{n} \ast \mu_{n})(j) + (1-p)[\bar{F}_{n}(j)^{2} - \bar{F}_{n}(j+1)^{2}],
\end{equation}
where
\begin{equation}
(\mu \ast \mu)(j) = \sum_{i=0}^{j} \mu(i)\,\mu(j-i)
\end{equation}
is the autoconvolution and
\begin{equation}
\bar{F}_{n}(j) = \PP(X_{n} \geq j) = \sum_{i \geq j} \mu_{n}(i)
\end{equation}
is the right-tail cdf. Since $\PP(\min(X,X')=j) = \bar{F}(j)^{2} - \bar{F}(j+1)^{2}$ for two \iid\ copies, the second term in (\ref{eq:full-rec}) is exactly the contribution from the `min' branch.

The recursion (\ref{eq:full-rec}) underlies a depth-saving choice of leaves: rather than start the simulations from all-ones leaves, it is advantageous to start from the $X_{0}(i)$ given by
\begin{equation}
\label{eq:onetwo}
X_{0}(i) = \begin{cases}
1, & \text{with probability } 1-p, \\
2, & \text{with probability } p,
\end{cases}
\end{equation}
which effectively enlarges the represented tree depth by one, or one level more, from
\begin{equation}
\label{eq:twice}
X_{0}(i) = \begin{cases}
1, & \text{with probability } (1-p)^{3}+2p(1-p)^{2}, \\
2, & \text{with probability } p(1-p)^{2}+p^{2}(1-p), \\
3, & \text{with probability } 2p^{2}(1-p), \\
4, & \text{with probability } p^{3},
\end{cases}
\end{equation}
so that a tree $T_{N}$ with the $X_{0}(i)$ given by (\ref{eq:twice}) is equivalent to a tree $T_{N+2}$ with all $X_{0}(i)=1$. One can move to still deeper levels at the cost of increasingly cumbersome expressions for the probabilities. These are the $K=1,2$ instances of precomputing the exact level-$K$ marginal $\mu_{K}$: iterated $K$ times from $\mu_{0}=\delta_{1}$, the recursion (\ref{eq:full-rec}) yields $\mu_{K}$, the law of the level-$K$ value in the original tree.

In an all-ones critical tree of depth $N$, the $2^{N-K}$ values found at level $K$ (counted from the leaves) are \iid\ from $\mu_{K}$, because the corresponding depth-$K$ subtrees share neither leaves nor internal-node operators. As a result, simulating the depth-$N$ process is statistically equivalent to (i)~precomputing $\mu_{K}$ once, then (ii)~simulating a depth-$(N-K)$ tree whose $2^{N-K}$ leaves are independently drawn from $\mu_{K}$. The total work per sample scales as
\begin{equation}
\label{eq:work}
W(N,K) \approx (c_{0}+c_{1}K)\,2^{N-K},
\end{equation}
with $c_{0}$ a small constant counting elementary operations per internal-node combine ($c_{0} \approx 5$ on our hardware) and $c_{1} \approx 1$ counting elementary operations per binary search through the leaf-table cdf (see Section~\ref{sec:appendix-simulator}).

A naive evaluation of (\ref{eq:full-rec}) at each step $k \to k+1$ costs $O((2^{k})^{2})$ floating-point operations, dominated by the autoconvolution; iterated $K$ times the total work is $O(4^{K})$. This is comfortable up to $K \approx 18$ (about $30$s on our hardware) but prohibitive past $K=22$. We replace the basic algorithm at larger $K$ by a radix-2 FFT autoconvolution~\cite{DanielsonLanczos1942,PressTeukolsky1992}: zero-pad to length $n_{\text{fft}}$, a power of two with $n_{\text{fft}} \geq 2(j_{\max}+1)$ to avoid circular wraparound, then FFT/square/inverse-FFT produces $(\mu \ast \mu)(j)$ on $[0, j_{\max}]$ in $O(n_{\text{fft}}\log{n_{\text{fft}}})$ operations. With cap $j_{\max} = 2^{m}-1$, the total cost across $K$ levels is $O(K\cdot 2^{m}\log{2^{m}})$; for $m=22$ and $K=35$, a few billion flops, running in $\approx 18$s on our hardware.

The cap $j_{\max} = 2^{m}-1$ truncates any mass the recursion would place at $j > j_{\max}$. By~\cite[Theorem~1]{AuffingerCable2017}, that mass is concentrated at $\log{j} > \sqrt{cK}$ in the limit at $p=1/2$, so $j_{\max} \gtrsim e^{\sqrt{cK}}$ keeps the truncation in the deep tail. We use $m=22$, giving $j_{\max} \approx 4.2 \times 10^{6}$: two orders of magnitude above the asymptotic right edge $e^{\sqrt{cK}} \approx 4 \times 10^{4}$ at $K=35$, and an order of magnitude above the largest sample observed in the production grid ($\max\XN = 4.1 \times 10^{5}$ at $N=60$). A matched-seed sensitivity check on the $N=48$ critical sweep run at $m=22$ versus $m=20$ finds $99{,}999$ of $10^{5}$ samples bit-identical and
\begin{equation}
\bigl|\widehat\EE[\log{\XN}/\sqrt{cN}]_{m = 22} - \widehat\EE[\log{\XN}/\sqrt{cN}]_{m = 20}\bigr| \approx 2.4 \times 10^{-10},
\end{equation}
six orders of magnitude below the Monte Carlo standard error $\sigma \approx 7.5 \times 10^{-4}$ of the rescaled mean. Truncation thus contributes negligibly to the reported moments and bulk quantiles; full-tail studies of the upper $10^{-5}$ region would require a larger buffer.

\subsection{The hybrid simulator}
\label{sec:appendix-simulator}

The simulator combines depth-first traversal of the upper tree with a tight level-wise inner loop on a per-thread reusable buffer. Pure depth-first uses $O(N)$ memory but pays one function call per internal node and is appreciably slower than the level-wise variant; pure level-wise allocates $O(2^{N})$ memory, which becomes the binding constraint past $N \approx 30$. The hybrid recurses depth-first down to a cutoff depth $K_{\text{cut}}$ (default $\min(N, 20)$) and then runs the level-wise inner loop on a $2^{K_{\text{cut}}} \times 8$\,B per-thread buffer. Force-inlining the level-wise function and marking the buffer pointer \texttt{restrict} are essential for the level-wise advantage to survive optimization under both Apple~clang and GCC.

The leaf draw branches three ways depending on the leaf-distribution mode: a fixed integer for the all-ones case; a constant-time table for the small-support modes (eqs.~(\ref{eq:onetwo}), (\ref{eq:twice}), and the Bernoulli variant); or an $O(\log_{2}{T})$ binary search through the cumulative thresholds of the precomputed leaf table, where $T$ is the number of nonzero entries in $\mu_{K}$. For $K=35$, $T \approx 4 \times 10^{6}$ and the search costs $\approx 22$ comparisons per leaf, i.e.\ $\approx 0.6$ per precompute level, consistent with the $c_{1}\approx 1$ coefficient of (\ref{eq:work}).

\subsection{Statistical bottlenecks}
\label{sec:appendix-statistics}

For a fixed threshold $x$, let $\pi_{x} = \PP(\XN \leq x)$. The Monte Carlo error on the empirical cdf value $\widehat{\PP}(\XN \leq x)$ is $\sigma_{x} \sim \sqrt{\pi_{x}(1-\pi_{x})/M} \leq 1/2\sqrt{M}$, independent of $N$. For the rescaled mean $\widehat{\EE}[\log{\XN}]/\sqrt{cN}$ relevant to the scaling test of Section~\ref{sec:critical}, the dependence on $N$ enters only through the variance of the limit law. By \cite[Theorem~1]{AuffingerCable2017}, the rescaled $\log{\XN}/\sqrt{cN}$ converges in distribution to $\Beta(2,1)$, which has variance $\sigma_{\mathcal{B}}^{2} = 1/18$. Consequently,
\begin{equation}
\label{eq:se-mean}
\sigma\bigl(\widehat{\EE}[\log{\XN}]/\sqrt{cN}\bigr) \approx \frac{\sigma_{\mathcal{B}}}{\sqrt{M}} = \frac{1}{\sqrt{18\,M}}.
\end{equation}
At $M=10^{5}$ this is about $7.5 \times 10^{-4}$, which resolves the theoretical limit $2/3$ to about $0.1\%$ relative error---enough to resolve the leading $1/\sqrt{N}$ correction at the depths we simulate.

Quantile errors inherit the $1/\sqrt{M}$ rate, with a constant set by the local density of $\Beta(2,1)$. For the median of $\log{\XN}/\sqrt{cN}$ the density is $2t$, giving $f(t_{0.5}) = \sqrt{2}$ and SE $\approx 1/\sqrt{8M}$, slightly worse than the SE on the mean. At $M=10^{5}$ the Monte Carlo error on the rescaled mean is already small compared with the finite-depth drift at $N \leq 60$; the main obstruction to a sharper extrapolation is the combination of finite-$N$ bias and the absence of a theoretically prescribed subleading correction. Larger $M$ would mainly help tail and extreme-quantile diagnostics; larger $N$ would be needed to stabilize extrapolations of the asymptote.

A linear extrapolation in $1/\sqrt{N}$ over the seven $K=35$ depths gives intercept $1.060 \pm 0.007$ for $\widehat{\EE}[\log{\XN}]/\sqrt{N}$, undershooting the theoretical $2\pi/3\sqrt{3} \approx 1.209$~\cite[Corollary~1]{AuffingerCable2017} by $\approx 12\%$. The gap is sensitive to the assumed subleading structure: adding a $d/N$ term lifts the intercept to $1.264 \pm 0.049$ with $d$ significant at $\sim 4\sigma$, while a non-parametric Bulirsch--Stoer rational extrapolation~\cite{BulirschStoer1964,PressTeukolsky1992} returns $1.18$. The three estimates (Table~\ref{tab:extrap}) straddle the asymptote within $\sim 10\%$, with the precise residual depending on the assumed next-order term---an unavoidable consequence of extrapolating from seven points without a theoretical anchor. The data are quantitatively consistent with the asymptote at the few-percent level once curvature is allowed; the subleading exponent is itself a research question.

\begin{table}[!htb]
\caption{\label{tab:extrap}Extrapolation of $\widehat{\EE}[\log{\XN}]/\sqrt{N}$ to $1/\sqrt{N} \to 0$ over the seven $K=35$ effective depths $N=36, 40, \ldots, 60$, against the theoretical asymptote $2\pi/3\sqrt{3}$ \cite{AuffingerCable2017}.}
\centering
\begin{tabular}{lcr}
\toprule
Method & Intercept & Residual \\
\midrule
OLS linear $a + b/\sqrt{N}$            & $1.060 \pm 0.007$ & $-12.4\%$ \\
OLS quadratic $a + b/\sqrt{N} + d/N$   & $1.264 \pm 0.049$ & $+4.5\%$  \\
Bulirsch--Stoer rational               & $1.18$            & $-2.0\%$  \\
\bottomrule
\end{tabular}
\end{table}


\section*{Acknowledgments}

The author thanks the anonymous referee of~\cite{Mendonca2020} who called his attention to~\cite{AuffingerCable2017}, which together with two seminars the author attended in Paris in 2019, given independently by B.~Derrida and Z.~Shi, on the Derrida--Retaux model sparked his interest in the min-plus process. The author also thanks the two anonymous reviewers of the present manuscript, one of whom suggested the near-critical experiment of Section~\ref{sec:nearcritical}. This work received partial financial support from Funda\c{c}\~{a}o de Amparo \`{a} Pesquisa do Estado de S\~{a}o Paulo -- FAPESP, Brazil, through grant no.~2020/04475-7.


\section*{Data and code availability}

The simulator, the leaf-table precompute tool, the run and analysis scripts, and the histograms and summary statistics of every simulation reported in this paper are openly available in a public GitHub repository and are archived on Zenodo under DOI \href{https://doi.org/10.5281/zenodo.21385814}{10.5281/zenodo.21385814}. Raw sample streams and the precomputed leaf tables regenerate exactly from the recorded seeds and scripts.


\medskip

\end{document}